\documentclass[runningheads]{llncs}
\usepackage[T1]{fontenc}
\usepackage{graphicx}
\usepackage[linesnumbered,ruled,commentsnumbered]{algorithm2e}
\SetKwComment{Comment}{/* }{ */}
\usepackage{amsmath, amssymb, amsfonts, latexsym, mathtools}
\usepackage{cleveref}

\begin{document}
\title{Reconfiguring Multiple Connected Components \\with Size Multiset Constraints}

\titlerunning{Reconfiguring Multiple Connected Components}
\author{Yu Nakahata\inst{1}\orcidID{0000-0002-8947-0994}}
\institute{Nara Institute of Science and Technology,\\ 
8916-5 Takayama-cho, Ikoma, Nara 630-0192, Japan\\
\email{yu.nakahata@is.naist.jp}}
\maketitle              % typeset the header of the contribution

\newcommand{\ISR}{\textsc{ISR}}
\newcommand{\CCR}{\textsc{CCR}}
\newcommand{\IMR}{\textsc{IMR}}

\newcommand{\TJ}{\textsf{TJ}}
\newcommand{\TS}{\textsf{TS}}
\newcommand{\CJ}{\textsf{CJ}}
\newcommand{\CS}{\textsf{CS}}
\newcommand{\CSone}{\textsf{CS1}}
\newcommand{\R}{\textsf{R}}

\newcommand{\compl}[1]{\overline{#1}}

\newcommand{\C}[1]{\mathcal{C}(#1)}
\renewcommand{\S}{\mathcal{S}}
\newcommand{\M}{\mathcal{M}}

\newcommand{\adj}[1]{\overset{#1}{\longleftrightarrow}}
\newcommand{\reach}[1]{\overset{#1}{\leftrightsquigarrow}}
\newcommand{\seq}[1]{\langle #1 \rangle}
\newcommand{\yes}{\longrightarrow}
\newcommand{\dist}[3]{d_{#1}(#2, #3)}

\newcommand{\set}[1]{\{#1\}}
\newcommand{\size}[1]{|#1|}
\newcommand{\weight}[1]{\size{\size{#1}}}
\newcommand{\edge}[1]{\set{#1}}
\newcommand{\ms}[1]{\{\!|#1|\!\}}

\newcommand{\order}[1]{\mathcal{O}(#1)}
\newcommand{\YES}{\textrm{YES}}
\newcommand{\NO}{\textrm{NO}}

\newcommand{\inv}[2]{\mathrm{inv}(#1, #2)}
\newcommand{\rank}[2]{\sigma_{#1}(#2)}

\newcommand{\piran}[2]{\Pi(#1, #2)}
\newcommand{\ccpiran}[2]{\Pi^{\mathtt{cc}}(#1, #2)}
\newcommand{\vcc}{V^{\mathtt{cc}}}
\newcommand{\ecc}{E^{\mathtt{cc}}}

\begin{abstract}
We propose a novel generalization of \textsc{Independent Set Reconfiguration} (\ISR): \textsc{Connected Components Reconfiguration} (\CCR).
In \CCR, we are given a graph $G$, two vertex subsets $A$ and $B$, and a multiset $\M$ of positive integers.
The question is whether $A$ and $B$ are reconfigurable under a certain rule, while ensuring that each vertex subset induces connected components whose sizes match the multiset $\mathcal{M}$.
\ISR\ is a special case of \CCR\ where $\M$ only contains 1.
We also propose new reconfiguration rules: \emph{component jumping} (\CJ) and \emph{component sliding} (\CS), which regard \emph{connected components as tokens}.
Since \CCR\ generalizes \ISR, the problem is PSPACE-complete.
In contrast, we show three positive results:
First, \CCR-\CS\ and \CCR-\CJ\ are solvable in linear and quadratic time, respectively, when $G$ is a path.
Second, we show that \CCR-\CS\ is solvable in linear time for cographs.
Third, when $\M$ contains only the same elements (i.e., all connected components have the same size), we show that \CCR-\CJ\ is solvable in linear time if $G$ is chordal.
The second and third results generalize known results for \ISR\ and exhibit an interesting difference between the reconfiguration rules.

\end{abstract}
\section{Introduction}
Imagine robots operating in a disaster area.
Now that the work is done, the robots are set to be reassigned to their new positions.
For safety reasons, we need to ensure that robots are not in the same area or adjacent to each other.
How should we move the robots to change their positions from the current arrangement to the target arrangement?
We assume the robots to be indistinguishable.

This situation can be modeled using \emph{combinatorial reconfiguration}~\cite{ito2011complexity}.
The aforementioned problem is \textsc{Independent Set Reconfiguration} (\ISR)~\cite{hearn2005pspace,ito2011complexity,kaminski2012complexity},  the most well-studied problem in the field.
In \ISR, we are given a graph $G$ and two vertex subsets, $A$ and $B$.
Imagine that a \emph{token} is placed on each vertex in $A$.
Can we move the tokens one by one to eventually lead to $B$ while keeping the tokens not in the same vertex or adjacent vertices?
In the above example, a robot corresponds to a token.

Let us generalize the above example to the following situation: Each robot has \emph{size} $k$ and occupies exactly $k$ vertices.
They can flexibly change their shapes like amoebas.
However, the $k$ vertices must be connected.
Additionally, robots should not occupy the same vertex or adjacent vertices for safety reasons.
How can we move the robots from the initial configuration to the target configuration?
We assume that robots of the same size are indistinguishable from each other.

To model the above problem, we define a new reconfiguration problem, \textsc{Connected Components Reconfiguration} (\CCR).
In \CCR, we are given a graph $G$, two vertex subsets $A$ and $B$, and a multiset $\M$ of positive integers.
The question is whether $A$ and $B$ are reconfigurable under some rule while keeping each vertex subset induces connected components whose sizes are equal to $\M$.
In the above example, $A$ and $B$ are the vertices occupied by robots in the initial and target configuration, respectively.
$\M$ describes the sizes of robots.
\ISR\ is a special case of \CCR\ where $\M$ only contains 1.

\emph{Reconfiguration rules} define the valid movement of tokens in reconfiguration problems.
The most well-studied reconfiguration rules are \emph{token jumping} (\TJ)~\cite{kaminski2012complexity} and \emph{token sliding} (\TS)~\cite{hearn2005pspace}.
\TS\ allows us to move a token to an adjacent vertex where no other token exists.
Additionally, \TJ\ enables us to move a token to a non-adjacent vertex if no other token exists at that vertex.
In the first example, moving a robot on the ground along an edge corresponds to \TS, while \TJ\ indicates that a robot can move to other distant vertices like a drone.

\begin{figure}[t]
    \centering
    \includegraphics[bb=0 0 420 217,width=0.8\linewidth]{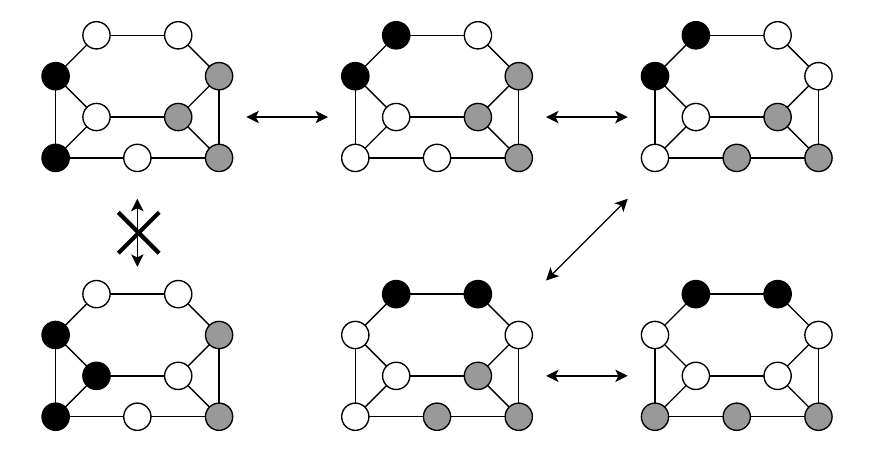}
    \caption{A reconfiguration sequence in \CCR-\CJ. The start and target configurations $A$ and $B$ are shown in the upper left and lower right, respectively.
    $\M$ consists of one 2 and one 3.
    Black and gray vertices are in vertex subsets.
    Note that the upper left configuration and the lower left configuration are not adjacent under \CJ\ because they exchange vertices between different connected components, which is allowed in \TJ\ and \TS.
    The reconfiguration sequence is also valid under \CS.}
    \label{fig:example}
\end{figure}

However, applying \TJ\ and \TS\ to \CCR\ causes an issue.
In \TJ\ and \TS, a token is a vertex.
Therefore, tokens are allowed to move between different robots, as illustrated on the left of \Cref{fig:example}.
We have assumed that the robots are flexible enough to change their shapes but not so flexible that they can self-repair.
Therefore, we need rules that make robots move while maintaining their sizes.
To this end, we propose two new reconfiguration rules: \emph{component jumping} (\CJ) and \emph{component sliding} (\CS), where \emph{a token is a connected component}.
When all the connected components have size 1, \CCR-\CJ\ and \CCR-\CS\ coincide with \ISR-\TJ\ and \ISR-\TS, respectively.\footnote{Here, \textsc{Prob}-\R\ indicates the problem \textsc{Prob} under the reconfiguration rule \R.}
An example of a reconfiguration sequence under \CJ\ and \CS\ is shown in \Cref{fig:example}.

Since \CCR-\CJ\ and \CCR-\CS\ generalize \ISR-\TJ\ and \ISR-\TS, respectively, they are PSPACE-complete in general.
In contrast, we show three positive results:
First, \CCR-\CS\ and \CCR-\CJ\ are solvable in linear and quadratic time, respectively, when $G$ is a path.
The former is easy, but the algorithm for the latter is non-trivial and has a connection to the sorting problem with a buffer.
Second, we show that \CCR-\CS\ is solvable in linear time for cographs.
Third, when $\M$ contains only the same elements (i.e., all connected components have the same size), we show that \CCR-\CJ\ is solvable in linear time if $G$ is chordal.
The second and third results generalize known results for \ISR\ and exhibit an interesting difference between the reconfiguration rules.

% Due to the space limitations,
Proofs of the theorems and lemmas marked with an asterisk (*) are shown in the Appendix.

\paragraph{Related work.}
Combinatorial reconfiguration~\cite{ito2011complexity} is an emerging area in theoretical computer science.
It is important from both practical and theoretical perspectives because many real-world systems are dynamic; additionally, it reveals the structure of the solution space.
Combinatorial reconfiguration is studied for several problems, including independent sets~\cite{hearn2005pspace,ito2011complexity,kaminski2012complexity},
dominating sets~\cite{haddadan2016complexity,suzuki2016reconfiguration},
and cliques~\cite{ito2023reconfiguration}.
See the survey for more details~\cite{nishimura2018introduction}.

Since \ISR\ is the most well-studied problem in combinatorial reconfiguration, problems that generalize \ISR\ have also been studied, such as \textsc{Regular Induced Subgraph Reconfiguration}~\cite{eto2022reconfiguration} and \textsc{Induced Isomorphic Subgraph Reconfiguration}~\cite{suga2025changing}.
Some studies generalize reconfiguration rules.
As generalizations of \TJ\ and \TS, Hatano et al.~\cite{hatano2025independent} proposed $d$-Jump, where one token can be moved to within a distance $d$, Suga et al.~\cite{suga2025changing} proposed $k$-\TJ\ and $k$-\TS, where at most $k$ tokens can move simultaneously, and, K{\v{r}}i{\v{s}}t'an and Svoboda~\cite{kvrivstan2025reconfiguration} proposed $(k, d)$-\TJ, where at most $k$ tokens can each move within a distance $d$.
Note that, however, a token is a vertex in these papers.
In this paper, we extend not only the problem (\ISR\ to \CCR) but also the token (a vertex to a connected component) and propose new reconfiguration rules \CJ\ and \CS.

\section{Preliminaries}\label{sec:pre}
In this paper, we consider a simple undirected graph $G = (V, E)$.
Throughout this paper, we denote the number of vertices by $n$.
For $U \subseteq V$, we define $N[U] = U \cup \set{v \in V : \exists u \in U, \edge{u, v} \in E}$, and $N(U) = N[U] \setminus U$.
% For a vertex $u \in V$, we use $N[u]$ and $N(u)$ for a shorthand of $N[\set{u}]$ and $N(\set{u})$, respectively.
For $U, U' \subseteq V$, \emph{$U$ touches $U'$} if $U \cup U'$ is connected.
Note that $U$ touches $U'$ if and only if $U'$ touches $U$.
For $u \in V$ and $U \subseteq V$, \emph{$u$ touches $U$} if $\set{u} \cup U$ is connected.
Let $G[U]$ be an induced subgraph by $U$.
When $G[U]$ has no edges, $U$ is an \emph{independent set}.
We say that $U$ is \emph{connected} if $G[U]$ is connected.
A subset of $U$ that is connected and maximal is a \emph{connected component} of $U$.
The \emph{size} of a connected component is the number of vertices in it.
Let $\C{U}$ be the set of connected components of $U$; that is, $\C{U}$ is a partition of $U$ and contains no empty set.
Let $m(U)$ be a multiset of sizes of connected components of $U$.
We refer to $m(U)$ as the \emph{CC-multiset} of $U$ in $G$.

For graphs $G$ and $H$, if $H$ is an induced subgraph of $G$, we say that \emph{$G$ contains $H$ as an induced subgraph}.
If $G$ does not contain $H$ as an induced subgraph, $G$ is \emph{$H$-free}.
For instance, \emph{cographs} are $P_4$-free graphs, where $P_4$ denotes a path with four vertices.
A graph $G$ is \emph{chordal} if $G$ is $C_\ell$-free for every $\ell \ge 4$, where $C_\ell$ denotes a cycle with $\ell$ vertices.
$G$ is \emph{even-hole-free} if $G$ is $C_\ell$-free for every even $\ell \ge 4$.
By definition, chordal graphs are even-hole-free.

\section{Our problem and reconfiguration rules}\label{sec:def}
In this section, we introduce \textsc{Connected Components Reconfiguration} (\CCR) and propose two new reconfiguration rules: \emph{component jumping} (\CJ) and \emph{component sliding} (\CS).

\begin{definition}[\CCR]\label{def:ccr}
    \textsc{Connected Components Reconfiguration} (\CCR) is defined as follows:
    \begin{description}
        \item[Input] A graph $G$, vertex subsets $A, B \subseteq V$, a multiset $\M$ consisting of positive integers, and a reconfiguration rule $\R$.
        \item[Output] Is there a sequence of vertex subsets from $A$ to $B$ where (1) every vertex subset has a CC-multiset equal to $\M$ and (2) every two consecutive vertex subsets are adjacent in \R?
    \end{description}
\end{definition}
If the answer is \YES, we say that \emph{$A$ and $B$ are reconfigurable} and call the sequence satisfying (1) and (2) a \emph{reconfiguration sequence from $A$ to $B$}.
The \emph{length} of a reconfiguration sequence $\seq{U_0 = A, U_1, \dots, U_{\ell} = B}$ is $\ell$.

Using the multiset $\M$, we can express the solution spaces of several problems: If $\M$ only contains 1, the solution space is the independent sets.
If $\M$ only contains 2, the solution space is the induced matchings.
In this way, \CCR\ generalizes the existing reconfiguration problems.

As for the reconfiguration rule $\R$, \TJ\ and \TS\ are the most well-studied ones in the literature~\cite{hearn2005pspace,ito2011complexity,kaminski2012complexity}.
A reconfiguration rule defines the adjacency relation between solutions.
The adjacency of $U, U' \subseteq V$ under a reconfiguration rule $\R$ is written as $U \adj{\R} U'$.
Then, \TJ\ and \TS\ are defined as follows.
\begin{description}
    \item[\TJ] $U \adj{\TJ} U' \Leftrightarrow |U \setminus U'| = |U' \setminus U| = 1$
    \item[\TS] $U \adj{\TS} U' \Leftrightarrow |U \setminus U'| = |U' \setminus U| = 1$ and $\{u, u'\} \in E(G)$, where $U \setminus U' = \set{u}, U' \setminus U = \set{u'}$
\end{description}
\TJ\ and \TS\ are often explained using the notion of \emph{token}.
Imagine that a token is placed on each vertex in $A$.
In each step, we can move one token to an unoccupied vertex; this is \TJ, and in \TS, we must move a token to its neighbor.
The reconfiguration problem asks: Can we move tokens from $A$ to $B$ under a reconfiguration rule while keeping the token configuration \emph{valid} (e.g., in \ISR, tokens are not adjacent)?

In addition to \TJ\ and \TS, we define two new reconfiguration rules: \emph{component jumping} (\CJ) and \emph
{component sliding} (\CS).
We also define a special case of \CS, \CSone, which only allows component sliding by one vertex.
In \TJ\ and \TS, a token is a vertex.
In contrast, in \CJ\ and \CS, \emph{a ``token'' is a connected component}.
In the following, if $|\C{U} \setminus \C{U'}| = |\C{U'} \setminus \C{U}| = 1$, we denote the only connected components in $\C{U} \setminus \C{U'}$ and $\C{U'} \setminus \C{U}$ by $C$ and $C'$, respectively.

\begin{definition}[\CJ, \CS, \CSone]\label{def:rule}
    For $U, U' \subseteq V(G)$ with $m(U) = m(U')$, the reconfiguration rules \CJ, \CS, and \CSone\ are defined as follows:
    \begin{description}
        \item[\CJ] $U \adj{\CJ} U' \Leftrightarrow |\C{U} \setminus \C{U'}| = |\C{U'} \setminus \C{U}| = 1$
        \item[\CS] $U \adj{\CS} U' \Leftrightarrow U \adj{\CJ} U'$ and $C \cup C'$ is connected
        \item[\CSone] $U \adj{\CSone} U' \Leftrightarrow U \adj{\CS} U'$ and $\size{C \setminus C'} = \size{C' \setminus C} = 1$
            \item 
    \end{description}
\end{definition}
\CJ\ and \CS\ are defined so that they generalize \TS\ and \TJ\ in \ISR, respectively.
Here, a token is individual but can change its body like an amoeba.
Each token has its \emph{size}, and a token with size $k$ occupies exactly $k$ vertices.
In addition, vertices occupied by a token must be connected.
In each step, \CJ\ allows us to move one token (connected component) with size $k$ to any vertex subset with $k$ vertices and does not touch the other connected components.
In \CS, in addition, the union of vertices occupied by the connected component before and after the movement must be connected.

Let us consider the relation between reconfiguration rules in \CCR.
We write $\R \yes \R'$ when: if the answer of \CCR-\R\ is \YES, the answer of \CCR-$\R'$ is also \YES.
The following theorem shows the relation between reconfiguration rules.
The summary is shown in \Cref{fig:relation}.

\begin{figure}[t]
    \centering
    \includegraphics[bb=0 0 202 109,width=0.5\linewidth]
    {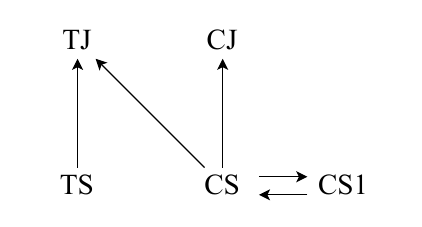}
    \caption{Relation between the reconfiguration rules. A solid arrow $\R \yes \R'$ means that, if the answer of \CCR-\R\ is \YES, the answer of \CCR-$\R'$ is also \YES.}
    % The dotted arrow holds if $\M$ does not contain two elements whose difference is 1, and the dashed arrow holds if, in addition, $\M$ does not contain 1.}
    \label{fig:relation}
\end{figure}

\begin{theorem}[*]\label{thm:relation}
    For a graph $G$ and two vertex subsets $U, U' \subseteq V$ with $m(U) = m(U')$, the following holds:
    \begin{enumerate}
        \item[(1)] $\CSone \yes \CS$ and $\CS \yes \CSone$,
        \item[(2)] $\TS \yes \TJ,\CS \yes \CJ$, and $\CS \yes \TJ$, 
        % \item[(3)] If $\M$ does not contain two elements whose difference is 1, $\TS \yes \CS$, and
        % \item[(4)] If $\M$ does not contain two elements whose difference is 1 and does not contain 1, $\TJ \yes \CS$.
    \end{enumerate}
\end{theorem}

\section{\CCR-\CJ\ and \CCR-\CS\ for path graphs}\label{sec:path}
In this section, we show that for path graphs, \CCR-\CS\ and \CCR-\CJ\ are solvable in linear and quadratic time, respectively.
Let $v_1, \dots, v_n$ be the vertices of a path graph $G$ aligned from left to right.
For $U \subseteq V$, let $h_{U}$ be the sequence of sizes of connected components along $G$ from left to right.
Using $h_{U}$, we can determine the reconfigurability under \CS.
\begin{lemma}\label{lem:path_cs}
    $A$ and $B$ are reconfigurable under \CS\ if and only if $h_{A} = h_{B}$.
\end{lemma}
\begin{proof}
    $\Leftarrow$) If $h_{A} = h_{B}$, both $A$ and $B$ can be reconfigurable into the \emph{left-most subset} $L$, where the vertices are kept to the left as much as possible.
    Since the reconfigurability is symmetric and transitive, it is possible to reconfigure from $A$ to $B$ by reconfiguring from $A$ to $L$ and from $L$ to $B$.

    $\Rightarrow$) We show the contraposition. When $h_{A}\neq h_{B}$, from $m(A) = m(B)$, there exists a size pair $x$ and $y$ $(x \neq y)$ whose positions are reversed at $h_{A}$ and $h_{B}$.
    Under \CS, these pairs cannot be swapped.
    This means that $A$ are $B$ not reconfigurable. \qed
\end{proof}

In the following, we show that not only the reconfigurability but also a reconfiguration sequence (if exists) can be computed in polynomial time.
To output a reconfiguration sequence efficiently, we use a \emph{compressed reconfiguration sequence}.
In the sequence, we identify the position of a connected component (path) by its leftmost vertex.
By outputting a pair of left-most positions of the connected component to be moved before and after each step, we can output a reconfiguration sequence efficiently.
\begin{theorem}\label{thm:path_cs}
    If $G$ is a path graph, \CCR-\CS\ is solvable in time $\order{n}$.
    If the answer is \YES, we can find a compressed reconfiguration sequence in length $\order{n^2}$ in time $\order{n^2}$.
\end{theorem}
\begin{proof}
    From \Cref{lem:path_cs}, output \YES\ if $h_{A} = h_{B}$, \NO\ otherwise. Since $h_{A}$ and $h_{B}$ can be computed in linear time, the algorithm runs in time $\order{n}$.
    If the answer is \YES, by reconfiguring $A$ to $L$ and then $L$ to $B$, we obtain a reconfiguration sequence.
    Since each connected component moves $\order{n}$ times, the length of the obtained reconfiguration sequence is $\order{n^2}$.
    \qed
\end{proof}
The length bound of a reconfiguration sequence in \Cref{thm:path_cs} is tight because there are instances for which any reconfiguration sequence requires length $\mathrm{\Omega}(n^2)$ in \ISR-\TS~\cite{demaine2015linear}.

Next, we consider \CJ.
In contrast to \CS, we can swap the positions of connected components when there is enough space.
For $U \subseteq V$ having $k$ connected components, we define the \emph{buffer} of $U$ in $G$ as $b(U, k) = n - |U| - k$.
Here, $|U| + k$ is the number of vertices occupied by $k$ connected components of $U$ when they are kept left.
The buffer means the number of the right vertices we can freely use when the vertices in $U$ are kept left.

Our algorithm for \CCR-\CJ\ imitates the \emph{bubble sort}.
However, our sequences may have duplicate elements.
Thus, we need some notations before stating the algorithm.
Let $h_A'$ be a sequence of $x_i$'s where $x$ is an element in $h_A$, and $i$ is the number of occurrences of $x$ in the prefix of $h_A$ until itself.
For instance, if $h_A = \seq{2, 2, 3, 3}$, then $h_A' = \seq{2_1, 2_2, 3_1, 3_2}$.
We refer to $x$ as an \emph{element} and $x_i$ as a \emph{subscripted element}.
In addition, we define $\sigma_A$ as a function from subscripted elements to integers.
Here, $\sigma_A(x_i)$ is the \emph{rank} of $x_i$ in $A$, that is, its index in $h_A$.
In the above example, $\sigma_A(3_1) = 3$ holds.
We similarly define $h_B'$ and $\sigma_B$.

The algorithm scans $h_A'$ from left to right, and if we find adjacent subscripted elements $x_i$ and $y_j$ such that their ranks are reversed in $h_B'$, then the algorithm tries to \emph{swap} the positions of $x_i$ and $y_j$.
However, since we can jump only one connected component in each step, we need a buffer with at least $\min\set{x, y}$.
Conversely, if this amount of buffer exists, we can swap $x_i$ and $y_j$ as follows:
Without loss of generality, we assume that $x < y$.
(Note that $x \neq y$ because, if so, $i < j$ holds and $\rank{A}{x_i} < \rank{A}{x_j}$ and $\rank{B}{x_i} < \rank{B}{x_j}$.)
First, we jump the connected component of size $x$ to the right buffer.
Second, we slide (actually jump) the connected component of size $y$ to the left.
Finally, we jump back the connected component of size $x$ from the right buffer to the left vacant space occupied by the connected component of size $y$.
We call this procedure the \emph{swap} of $x$ and $y$.

To establish the necessary and sufficient conditions for reconfigurability, we define an \emph{inversion} between $h_A'$ and $h_B'$ as a pair of subscripted elements $(x_i, y_j)$ such that the rank of $x_i$ is smaller than $y_j$ in $h_A'$, but the reverse holds for $h_B'$.
We denote the set of inversions by $\inv{h_A'}{h_B'}$, that is, $\inv{h_A'}{h_B'} = \set{(x_i, y_j) : \rank{A}{x_i} < \rank{A}{y_i}, \rank{B}{x_i} > \rank{B}{y_i}}$.
The next lemma provides a necessary and sufficient condition for reconfigurability under \CJ\ with respect to inversions.
\begin{lemma}\label{lem:buffer}
    $A$ and $B$ are reconfigurable under \CJ\ if and only if 
    \begin{equation}\label{eq:buffer}
        \max_{(x_i, y_j) \in \inv{h_A'}{h_B'}}
            \min \set{x, y}\ \le\ b(A, k). 
    \end{equation}
\end{lemma}
\begin{proof}
    $\Rightarrow$) For each inversion $(x_i, y_j)$, to swap their positions, we need a buffer with $\min \set{x, y}$. Thus, the only-if direction holds.
    
    $\Leftarrow$) If Inequality~\eqref{eq:buffer} holds, we can \emph{sort $h_A'$ to $h_B'$} using the bubble sort.
    While there is an inversion, we repeat the following procedure:
    We scan $h_A'$ from left to right, and if we find adjacent subscripted elements $(x_i, y_j)$ such that they are an inversion, we swap them by using the right buffer.
    The correctness follows from the proof of bubble sort.
    Once $h_A'$ equals $h_B'$, by \Cref{lem:path_cs}, $A$ and $B$ are reconfigurable under \CS\ and thus under \CJ.
    \qed
\end{proof}

\begin{theorem}\label{thm:path_cj}
    If $G$ is a path graph, \CCR-\CJ\ is solvable in time $\order{n^2}$.
    If the answer is \YES, there is a compressed reconfiguration sequence of length $\order{n^2}$ and we can output the sequence in time $\order{n^2}$.
\end{theorem}
\begin{proof}
    By \Cref{lem:buffer}, the answer is \YES\ if Inequality~\eqref{eq:buffer} holds and \NO\ otherwise.
    The algorithm first computes $h_A'$ and $h_B'$ in time $\order{n}$.
    Second, we compute inversions in time $\order{n^2}$.
    Then, Inequality~\eqref{eq:buffer} can be checked in the same time bound.
    To obtain a reconfiguration sequence, we first make $A$ and $B$ leftmost in $\order{n}$ steps.
    Then, we execute the bubble sort from $h_A'$ to $h_B'$ as shown in the proof of \Cref{lem:buffer}.
    Each swap consists of three steps; hence, the total number of steps is $\order{n^2}$.
    We obtain a reconfiguration sequence from $A$ to $B$ by concatenating the above sequences.
    We achieve the claimed time bound by using the compressed reconfiguration sequence described right before \Cref{thm:path_cs}.
    \qed
\end{proof}

\section{\CCR-\CS\ for cographs}\label{sec:cograph}
In this section, we show a linear-time algorithm for \CCR-\CS\ when $G$ is a cograph.
If the answer is \YES, we can also find a \emph{shortest} reconfiguration sequence of length $\order{|V|}$ under \CS\ or \CSone.
Note that reconfigurability is equivalent for \CS\ and \CSone, but the lengths of shortest reconfiguration sequences may differ.
Our algorithm is based on that for \ISR-\TS\ on a cograph by Kami{\'n}ski et al.~\cite{kaminski2012complexity}.
Note that \ISR-\TS\ is equivalent to \CCR-\CS\ when $\M$ only contains 1.

We prepare some notations.
The \emph{complement graph} $\compl{G}$ of a graph $G = (V, E)$ is defined as $\compl{G} = (V, \set{\edge{u, v} : u, v \in V, \edge{u, v} \notin E})$.
A graph $G$ is a cograph if and only if, for every induced subgraph $F$ of $G$ with at least two vertices, either $F$ or $\compl{F}$ is disconnected~\cite{corneil1981complement}.
A \emph{co-component} of a graph $G$ is the subgraph induced by the vertex set of a connected component in $\compl{G}$.
Note that, if $G$ is a connected cograph, $V(G)$ is partitioned into vertex sets each of which induces a co-component of $G$.

The algorithm solves the problem in divide-and-conquer manner using a \emph{cotree}~\cite{corneil1985linear}, a decomposition tree of a cograph with respect to taking components and co-components.
As one of the base cases, we use the next two lemmas.
In the following, for $X, Y \subseteq V$ and a reconfiguration rule \R, we use $\dist{\R}{X}{Y}$ as the shortest distance (length of a shortest reconfiguration sequence) between $X$ and $Y$ under $\R$.

\begin{lemma}\label{lem:cograph_CS}
    Let $G$ be a connected cograph and $X, Y$ be connected vertex subsets with $|X| = |Y|$. Then, the following holds.
    \begin{equation}
        \dist{\CS}{X}{Y} = 
        \left\{
            \begin{array}{ll}
            0 & \quad (X = Y) \\
            1 & \quad (X \neq Y \mbox{ and } X \mbox{ touches } Y) \\
            2 & \quad (X \mbox{ does not touch } Y)
            \end{array}
        \right.
    \end{equation}
\end{lemma}
\begin{proof}
    The first case is obvious.  
    If $X \neq Y$ and $X$ touches $Y$, then $\dist{\CS}{X}{Y} \ge 1$.  
    Since $X \cup Y$ is connected, $X \adj{CS} Y$ and hence $\dist{\CS}{X}{Y} = 1$.  
    Otherwise, $X$ does not touch $Y$, so $\dist{\CS}{X}{Y} \ge 2$.  
    As $G$ is $P_4$-free, there exists a vertex $z$ adjacent to both $X$ and $Y$.  
    Let $Z$ be a connected vertex subset of size $|X|$ containing $z$.  
    Then $X \adj{CS} Z$ and $Z \adj{CS} Y$, so $\dist{\CS}{X}{Y} = 2$. \qed
\end{proof}

\begin{lemma}[*]\label{lem:cograph_CSone}
    Let $G$ be a connected cograph and $X, Y$ be connected vertex subsets with $|X| = |Y|$. If $X$ touches $Y$, $\dist{\CSone}{X}{Y} = \size{X \setminus Y}$.
    Otherwise, $\dist{\CSone}{X}{Y} = \size{X \setminus Y} + 1$.
\end{lemma}

\begin{theorem}\label{thm:cograph}
    \CCR-\CS\ is solvable in time $\order{\size{V} + \size{E}}$ if the input graph $G = (V, E)$ is a cograph.
    If the answer is \YES, we can compute a shortest reconfiguration sequence under \CS\ or \CSone\ in time $\order{\size{V} + \size{E}}$ and its length is $\order{|V|}$.
\end{theorem}
\begin{proof}
    Our algorithm is based on Kami{\'n}ski et al.~\cite{kaminski2012complexity}.
    The differences are:
    \begin{itemize}
        \item The first pruning condition $\size{A} \neq \size{B}$ is replaced by $m(A) \neq m(B)$.
        \item The base case $\size{A} = \size{B} = 1$ is replaced by $\size{\C{A}} = \size{\C{B}} = 1$. In addition, the procedure for this base case is substituted by Lemmas~\ref{lem:cograph_CS} or \ref{lem:cograph_CSone}.
    \end{itemize}
    Intuitively, these differences correspond to the notion of a token: a token changed from a vertex to a connected component.
    (We show pseudocode in \Cref{app:cograph} for reference, but the following proof is self-contained.)

    If $m(A) \neq m(B)$, the answer is \NO.
    If $\size{V(G)} = 1$, the problem is trivial.
    We assume that $G$ has at least two vertices in the following.
    Now, either $G$ or $\compl{G}$ is disconnected from the above characterization of cographs.
    If $G$ is disconnected, let $C_1, \dots, C_k$ be the connected components of $G$.
    We can solve the problem recursively for the connected components $C_1, \dots, C_k$ with respective vertex subsets $(A \cap C_1, B \cap C_1), \dots, (A \cap C_k, B \cap C_k)$.
    If one of the outputs is \NO, then the answer is NO.
    Otherwise, the answer is \YES\, and we merge subsequences to obtain a shortest reconfiguration sequence from $A$ to $B$.

    If $\compl{G}$ is disconnected, $G = (V, E)$ is connected.
    If in addition $\size{\C{A}} = \size{\C{B}} = 1$, by \Cref{lem:cograph_CS,lem:cograph_CSone}, the answer is \YES, and there exists a shortest reconfiguration sequence of $\order{|V|}$ from $A$ to $B$.
    Otherwise, if $A$ and $B$ are contained in the same co-component of $G$, we solve the problem for $A$ and $B$ recursively on that co-component.
    Otherwise, the answer is \NO\ because we cannot move any vertex in $A$ to the other co-component.
    Since there is an edge for every two vertices of different co-components, if we move a vertex outside the co-component, some connected components in $A$ will be merged.
    The correctness of the above procedure follows from \Cref{lem:cograph_CS,lem:cograph_CSone} and the characterizations of cographs.

    Given a cograph $G$, it is known that a \emph{cotree}, a decomposition tree with respect to taking components and co-components, can be computed in time $\order{\size{V} + \size{E}}$~\cite{corneil1985linear}.
    Using this, the above algorithm can be implemented in time $\order{\size{V} + \size{E}}$.
    The total length of the reconfiguration sequence is $\order{|V|}$ with respect to the input graph because we divide the problem into disjoint ones.\qed
\end{proof}

\section{\CCR-\CJ\ for connected components of equal size}\label{sec:chordal}
In this section, we consider \CCR-\CJ\ where $A$ and $B$ contain only connected components of equal size.
\ISR-\TJ\ is its special case where the size of each connected component is 1.
Our main tool is the \emph{CC-Piran graph}, a generalization of the Piran graph used by Kami{\'n}ski et~al.~\cite{kaminski2012complexity}.
They have shown that, given a graph $G$ and two independent sets $A$ and $B$, \ISR-\TJ\ is solvable in linear time if the Piran graph defined by $G, A, B$ is even-hole-free.
Analogously to their results, we show that given a graph $G$ and two vertex subsets $A, B$ whose all connected components have the same size, if the CC-Piran graph defined by $G, A, B$ is even-hole-free, \CCR-\CJ\ is solvable in linear time.

For two independent sets $A$ and $B$, the Piran graph $\piran{A}{B}$ of $A$ and $B$ is the subgraph of $G$ induced by the vertex set $(A \setminus B) \cup (B \setminus A)$.
We extend this definition to the CC-Piran graph.
\begin{definition}[CC-Piran graph]
    For a graph $G$ and $A, B \subseteq V(G)$, the CC-Piran graph $\ccpiran{A}{B} = (\vcc, \ecc)$ is defined as follows:
    \begin{itemize}
        \item $\vcc = (\C{A} \setminus \C{B}) \cup (\C{B} \setminus \C{A})$
        \item $\ecc = \set{\edge{C_A, C_B} : C_A \in \C{A} \setminus \C{B}, C_B \in \C{B} \setminus \C{A}, C_A \mbox{ touches } C_B}$
        % N[C_A] \cap C_B \neq \emptyset}$
    \end{itemize}
\end{definition}

Note that the CC-Piran graph is equivalent to the Piran graph if $A$ and $B$ are independent sets.
In the following, to avoid confusion, we call a vertex in the CC-Piran graph a \emph{component}.

\begin{theorem}\label{thm:even}
    Let $G$ be a graph and $A$ and $B$ be vertex subsets that only contain connected components of the same size.
    If the CC-Piran graph $\ccpiran{A}{B}$ is even-hole-free, then $A$ and $B$ are reconfigurable under $\CJ$.
    Moreover, there exists an algorithm running in time $\order{\size{V} + \size{E}}$ (if the CC-Piran graph is even-hole-free) that finds a shortest reconfiguration sequence from $A$ to $B$.
\end{theorem}
\begin{proof}
    We briefly review the idea of Kami{\'n}ski et al.~\cite{kaminski2012complexity}.
    The Piran graph $\piran{A}{B}$ is bipartite, and as such it does not contain odd cycles. If, in addition, $\piran{A}{B}$ is also even-hole-free, the graph must be a forest.
    Since $\size{A \setminus B} = \size{B \setminus A}$, by analyzing the number of edges, there exists a vertex in $B \setminus A$ with at most one neighbor in $A \setminus B$.
    Using this property, a simple greedy algorithm works: Find a vertex $v$ from $B \setminus A$ with at most one neighbor in $A \setminus B$.
    If $v$ has a neighbor in $A \setminus B$, say $w$, jump $w$ to $v$.
    Otherwise, jump an arbitrary token $w$ from $A \setminus B$ to $v$.
    Replace $A$ and $B$ with $A \setminus \set{w}$ and $B \setminus \set{v}$, respectively. While $\size{A} \ge 1$, repeat the procedure.
    Since any reconfiguration sequence under \TJ\ requires $\size{A \setminus B}$ steps, the obtained sequence is shortest.
    
    We show that the same algorithm works for our case.
    Since $\ccpiran{A}{B}$ is bipartite and even-hole-free, the graph must be a forest.
    Since $\size{\C{A} \setminus \C{B}} = \size{\C{B} \setminus \C{A}}$, there exists a component in $\C{B} \setminus \C{A}$ with at most one neighbor in $\C{A} \setminus \C{B}$.
    Thus, the above algorithm also works for $\ccpiran{A}{B}$.
    The correctness is preserved because each token has the same size and any token of $\C{A}$ can be matched with any other of $\C{B}$.
    Note that a ``token'' is a component in our case while it is a vertex in the previous case.
    
    As for the time complexity, we show that given a graph $G$ and vertex subsets $A$ and $B$, the CC-Piran graph $\ccpiran{A}{B}$ can be constructed in linear time.
    Assuming that the graph is given with adjacency lists, we construct $\ccpiran{A}{B}$ in two scans of the lists.
    First, we compute the connected components of $A$ and $B$.
    Second, we compute the adjacency relation between components. \qed
\end{proof}
    
Note that \Cref{thm:even} does not say anything about the input graph class.
The result of Kami{\'n}ski et~al.~\cite{kaminski2012complexity} indicates that \ISR-\TJ\ is in P if $G$ is even-hole-free because even-hole-freeness is closed under taking induced subgraphs: when $G$ is even-hole-free, every induced subgraph of $G$ (and thus the Piran graph) is even-hole-free.
In contrast, the CC-Piran graph allows us to take a \emph{minor} of $G$, which leads to the contraction of edges.
Although even-hole-freeness is closed under taking induced subgraphs, it is not closed under taking minors because contracting an edge in an odd hole may create an even hole.

\begin{figure}[t]
    \centering
    \includegraphics[bb=0 0 228 139,width=0.5\linewidth]{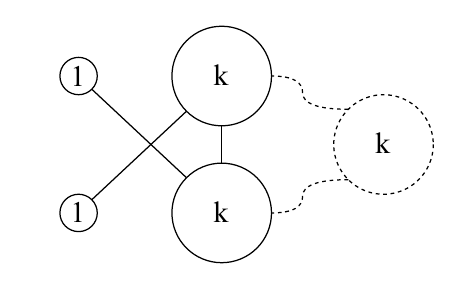}
    \caption{The CC-Piran graph such that the greedy algorithm fails.
    The solid upper and lower circles are components in $\C{A} \setminus \C{B}$ and $\C{B} \setminus \C{A}$, respectively.
    The dotted circle indicates a connected component with $k$ vertices in $G$ outside the CC-Piran graph.
    In this instance, we can reconfigure $A$ into $B$ by first moving the token of size $k$ to the right space, next moving the token of size 1, and then moving the token of size $k$ to the lower right space.
    }
    \label{fig:piran}
\end{figure}

In addition, as demonstrated in the proof of \Cref{thm:even}, the assumption that all connected components have the same size is important for the greedy algorithm works.
If there exist connected components with different sizes, there is a case we need a ``detour'', as illustrated in \Cref{fig:piran}.

Then, our next question is: For what graph $G$, the CC-Piran graph is even-hole-free?
The next lemma answers the question.
\begin{lemma}[*]\label{lem:chordal}
    If $G$ is chordal, the CC-Piran graph is even-hole-free.
\end{lemma}

We obtain the following corollary. Claimed time complexity is achieved by outputting a reconfiguration sequence using \TJ\ on the CC-Piran graph.

\begin{corollary}\label{cor:chordal}
    If $G$ is chordal and $A, B \subseteq V$ consists of only connected components of the same size, \CCR-\CJ\ is solvable in time $\order{\size{V} + \size{E}}$.
    In addition, we can find a shortest reconfiguration sequence in the same time bound.
\end{corollary}

\Cref{cor:chordal} exhibits an interesting contrast between \CJ\ and \CS: When $\M$ contains only 2, which corresponds to \textsc{Induced Matching Reconfiguration} (\IMR), is known to be PSPACE-complete under \TJ\ for chordal graphs~\cite{eto2022reconfiguration}.
% By \Cref{thm:relation}, the problem is also PSPACE-complete under \CS.
\Cref{cor:chordal} shows that, however, \IMR\ is solvable in linear time for chordal graphs under \CJ.
% This indicates that reconfiguration rules largely affect computational complexity.
This indicates that what we regard as a token largely affects computational complexity.

\section{Concluding remarks}\label{sec:conclusion}
Future work includes further analysis of the computational complexity of \CCR.
Some open problems are listed below:
\begin{itemize}
    \item Are \CCR-\CJ\ and \CCR-\CS\ tractable for trees?
    \item Is \CCR-\CJ\ on cographs in P? \ISR-\TJ\ is in P for cographs~\cite{bonamy2014reconfiguring,bonsma2016independent}.
    \item Is \CCR-\CJ\ tractable for even-hole-free graphs with general token sizes?
    \item Are \CCR-\CJ\ and \CCR-\CS\ in XP (i.e., solvable in time $n^{f(k)}$ for some computable function $f$) when parameterized by the number $k$ of connected components? Since \ISR-\TJ\ and \ISR-\TS\ are W[1]-hard when parameterized by $k$~\cite{ito2014parameterized,lokshtanov2018reconfiguration}, there unlikely exist FPT time (i.e., $f(k) n^{\order{1}}$) algorithms.
\end{itemize}

\subsubsection*{Acknowledgment}
% We thank anonymous reviewers for valuable comments.
This work is partially supported by JSPS KAKENHI JP22K17851, JP23K24806, and JP24K02931.


\bibliographystyle{splncs04}
\bibliography{reference}

\begin{thebibliography}{10}
\providecommand{\url}[1]{\texttt{#1}}
\providecommand{\urlprefix}{URL }
\providecommand{\doi}[1]{https://doi.org/#1}

\bibitem{bonamy2014reconfiguring}
Bonamy, M., Bousquet, N.: Reconfiguring independent sets in cographs. arXiv preprint arXiv:1406.1433  (2014)

\bibitem{bonsma2016independent}
Bonsma, P.: Independent set reconfiguration in cographs and their generalizations. Journal of Graph Theory  \textbf{83}(2),  164--195 (2016)

\bibitem{corneil1981complement}
Corneil, D.G., Lerchs, H., Burlingham, L.S.: Complement reducible graphs. Discrete Applied Mathematics  \textbf{3}(3),  163--174 (1981)

\bibitem{corneil1985linear}
Corneil, D.G., Perl, Y., Stewart, L.K.: A linear recognition algorithm for cographs. SIAM Journal on Computing  \textbf{14}(4),  926--934 (1985)

\bibitem{demaine2015linear}
Demaine, E.D., Demaine, M.L., Fox-Epstein, E., Hoang, D.A., Ito, T., Ono, H., Otachi, Y., Uehara, R., Yamada, T.: Linear-time algorithm for sliding tokens on trees. Theoretical Computer Science  \textbf{600},  132--142 (2015)

\bibitem{elbassioni2015polynomial}
Elbassioni, K.: A polynomial delay algorithm for generating connected induced subgraphs of a given cardinality. Journal of Graph Algorithms and Applications  \textbf{19}(1),  273--280 (2015)

\bibitem{eto2022reconfiguration}
Eto, H., Ito, T., Kobayashi, Y., Otachi, Y., Wasa, K.: Reconfiguration of regular induced subgraphs. In: International Conference and Workshops on Algorithms and Computation. pp. 35--46. Springer (2022)

\bibitem{haddadan2016complexity}
Haddadan, A., Ito, T., Mouawad, A.E., Nishimura, N., Ono, H., Suzuki, A., Tebbal, Y.: The complexity of dominating set reconfiguration. Theoretical Computer Science  \textbf{651},  37--49 (2016)

\bibitem{hatano2025independent}
Hatano, H., Kitamura, N., Izumi, T., Ito, T., Masuzawa, T.: Independent set reconfiguration under bounded-hop token jumping. In: International Conference and Workshops on Algorithms and Computation. pp. 215--228. Springer (2025)

\bibitem{hearn2005pspace}
Hearn, R.A., Demaine, E.D.: {PSPACE}-completeness of sliding-block puzzles and other problems through the nondeterministic constraint logic model of computation. Theoretical Computer Science  \textbf{343}(1-2),  72--96 (2005)

\bibitem{ito2011complexity}
Ito, T., Demaine, E.D., Harvey, N.J., Papadimitriou, C.H., Sideri, M., Uehara, R., Uno, Y.: On the complexity of reconfiguration problems. Theoretical Computer Science  \textbf{412}(12-14),  1054--1065 (2011)

\bibitem{ito2014parameterized}
Ito, T., Kami{\'n}ski, M., Ono, H., Suzuki, A., Uehara, R., Yamanaka, K.: On the parameterized complexity for token jumping on graphs. In: International Conference on Theory and Applications of Models of Computation. pp. 341--351. Springer (2014)

\bibitem{ito2023reconfiguration}
Ito, T., Ono, H., Otachi, Y.: Reconfiguration of cliques in a graph. Discrete Applied Mathematics  \textbf{333},  43--58 (2023)

\bibitem{kaminski2012complexity}
Kami{\'n}ski, M., Medvedev, P., Milani{\v{c}}, M.: Complexity of independent set reconfigurability problems. Theoretical computer science  \textbf{439},  9--15 (2012)

\bibitem{kvrivstan2025reconfiguration}
K{\v{r}}i{\v{s}}t'an, J.M., Svoboda, J.: Reconfiguration using generalized token jumping. In: International Conference and Workshops on Algorithms and Computation. pp. 244--265. Springer (2025)

\bibitem{lokshtanov2018reconfiguration}
Lokshtanov, D., Mouawad, A.E., Panolan, F., Ramanujan, M., Saurabh, S.: Reconfiguration on sparse graphs. Journal of Computer and System Sciences  \textbf{95},  122--131 (2018)

\bibitem{nishimura2018introduction}
Nishimura, N.: Introduction to reconfiguration. Algorithms  \textbf{11}(4), ~52 (2018)

\bibitem{suga2025changing}
Suga, T., Suzuki, A., Tamura, Y., Zhou, X.: Changing induced subgraph isomorphisms under extended reconfiguration rules. In: International Conference and Workshops on Algorithms and Computation. pp. 346--360. Springer (2025)

\bibitem{suzuki2016reconfiguration}
Suzuki, A., Mouawad, A.E., Nishimura, N.: Reconfiguration of dominating sets. Journal of Combinatorial Optimization  \textbf{32},  1182--1195 (2016)

\end{thebibliography}

\newpage

\appendix

\section{Omitted proof in \Cref{sec:def}}
In this section, we show the proof of \Cref{thm:relation}.
We use the following lemma~\cite{elbassioni2015polynomial}.

\begin{lemma}[\cite{elbassioni2015polynomial}]\label{lem:one_cc}
    Let $G$ be a connected graph and $X, Y \subseteq V$ be connected vertex subsets with $|X| = |Y|$.
    Then there exists a sequence $\seq{U_0 = X, U_1, \dots, U_\ell = Y}$ of connected vertex subsets such that, for all $i \in \set{0, 1, \dots, \ell - 1}$, $\size{U_i \setminus U_{i+1}} = \size{U_{i+1} \setminus U_i} = 1$.
\end{lemma}

Now we show the proof of \Cref{thm:relation}.

% \begin{theorem}[*]\label{thm:relation}
%     For a graph $G$ and two vertex subsets $U, U' \subseteq V$ with $m(U) = m(U')$, the following holds:
%     \begin{enumerate}
%         \item[(1)] $\CSone \yes \CS$ and $\CS \yes \CSone$,
%         \item[(2)] $\TS \yes \TJ,\CS \yes \CJ$, and $\CS \yes \TJ$, 
%         \item[(3)] If $\M$ does not contain two elements whose difference is 1, $\TS \yes \CS$, and
%         \item[(4)] If $\M$ does not contain two elements whose difference is 1 and does not contain 1, $\TJ \yes \CS$.
%     \end{enumerate}
% \end{theorem}
\begin{proof}
    In the following, let $U, U' \subseteq V$ with $m(U) = m(U')$.
    
    (1) $\CSone \yes \CS$ follows from the definition.
    We show $\CS \yes \CSone$.
    Let $\C{U} \setminus \C{U'} = \set{C}$ and $\C{U'} \setminus \C{U} = \set{C'}$.
    Note that $|C| = |C'|$ from $m(U) = m(U')$.
    If $|C| = |C'| = 1$, then $\size{C \setminus C'} = \size{C' \setminus C} = 1$, and thus $U \adj{\CSone} U'$.
    We assume $|C| = |C'| \ge 2$ in the following.
    Since $C$ and $C'$ do not touch the other components in $\C{U} \cap \C{U'}$, $C \cup C'$ does not touch any component in $\C{U} \cap \C{U'}$.
    Since $C \cup C'$ is connected, from \Cref{lem:one_cc}, there exists a sequence of vertex subsets $\seq{U_0 = C, U_1, \dots, U_\ell = C'}$ of $G[C \cup C']$ such that, for all $i \in \set{0, 1, \dots, \ell - 1}$, $\size{U_i \setminus U_{i+1}} = \size{U_{i+1} \setminus U_i} = 1$.
    For $i \in \set{0, 1, \dots, \ell}$, let $W_i = U_i \cup (U \cap U')$.
    Then, the sequence $\seq{W_0, W_1, \dots, W_\ell}$ is a reconfiguration sequence from $U$ to $U'$ under $\CSone$.
    This is because $W_0 = U, W_\ell = U'$, $m(W_i) = m(U)$ for all $i \in \set{0, 1, \dots, \ell}$, and $\size{U_i \setminus U_{i+1}} = \size{U_{i+1} \setminus U_i} = 1$ for all $i \in \set{0, 1, \dots, \ell - 1}$.

    (2) $\TS \yes \TJ$ and $\CS \yes \CJ$ are clear from the definition.
    We show $\CS \yes \TJ$.
    Since $\CS \yes \CSone$, it suffices to show $\CSone \yes \TJ$.
    This is true because, if $U \adj{\CSone} U'$, then $\size{U \setminus U'} = \size{U' \setminus U} = 1$ holds.

\end{proof}

% \section{Omitted proofs in \Cref{sec:path}}
% \subsection{Proof of \Cref{lem:path_cs}}
% \begin{proof}
%     $\Leftarrow$) If $h_{A} = h_{B}$, both $A$ and $B$ can be reconfigurable into the \emph{left-most subset} $L$, where the vertices are kept to the left as much as possible.
%     Since the reconfigurability is symmetric and transitive, it is possible to reconfigure from $A$ to $B$ by reconfiguring from $A$ to $L$ and from $L$ to $B$.

%     $\Rightarrow$) We show the contraposition. When $h_{A}\neq h_{B}$, from $m(A) = m(B)$, there exists a size pair $x$ and $y$ $(x \neq y)$ whose positions are reversed at $h_{A}$ and $h_{B}$.
%     Under \CS, these pairs cannot be swapped.
%     This means that $A$ are $B$ not reconfigurable. \qed
% \end{proof}

% \subsection{Proof of \Cref{thm:path_cs}}
% \begin{proof}
%     From \Cref{lem:path_cs}, output \YES\ if $h_{A} = h_{B}$, \NO\ otherwise. Since $h_{A}$ and $h_{B}$ can be computed in linear time, the algorithm runs in time $\order{n}$.
%     If the answer is \YES, by reconfiguring $A$ to $L$ and then $L$ to $B$, we obtain a reconfiguration sequence.
%     Since each connected component moves $\order{n}$ times, the length of the obtained reconfiguration sequence is $\order{n^2}$.
%     \qed
% \end{proof}

\section{Omitted proof and pseudocode in \Cref{sec:cograph}}\label{app:cograph}

\subsection{Proof of \Cref{lem:cograph_CSone}}
\begin{proof}
    We use induction on $|X|\ (= |Y|)$.
    As base cases, we consider when $|X| = |Y| = 1$.
    If $X = Y$, then $\dist{\CSone}{X}{Y} = 0 = \size{X \setminus Y}$.
    Otherwise, let $X = \set{x}$ and $Y = \set{y}$.
    If $X$ touches $Y$, $\dist{\CSone}{X}{Y} = 1 = \size{X \setminus Y}$ because $x$ and $y$ are adjacent.
    If $X$ does not touch $Y$, the shortest distance between $x$ and $y$ in $G$ is 2 because $G$ is $P_4$-free.
    Therefore, $\dist{\CSone}{X}{Y} = 2 = \size{X \setminus Y} + 1$.

    For induction steps, we consider several cases.
    In the following, we assume that $|X| = |Y| \ge 2$.
    \emph{Case~A}: $X$ or $Y$ intersects with multiple co-components.
    Without loss of generality, we assume that $X$ intersects with multiple co-components.
    Then, $X$ touches $Y$, and thus there exists $x \in X$ such that $x$ touches $Y$.
    We move $x' \in X \setminus \set{x}$ to any $y \in Y \setminus X$.
    Now the instance $(X, Y)$ is reduced to $(X' = X \setminus \set{x'}, Y' = Y \setminus \set{y})$.
    The reduced case is in Case~A, and thus $\dist{\CSone}{X}{Y} \le 1 + \size{X' \setminus Y'} = |X \setminus Y|$.
    Since $\dist{\CSone}{X}{Y} \ge |X \setminus Y|$, we obtain $\dist{\CSone}{X}{Y} = |X \setminus Y|$.
    
    In the following cases, we assume that $X$ and $Y$ are contained in co-components $C_X$ and $C_Y$, respectively.
    \emph{Case~B}: $C_X \neq C_Y$.
    In this case, $X \cap Y = \emptyset$ and $X$ touches $Y$.
    Take arbitrarily $x \in X$ and $y \in Y$.
    Then, $(X \setminus \set{x}) \cup \set{y}$ is connected because every $x' \in X$ and $y' \in Y$ are adjacent.
    By moving $x$ to $y$, the instance $(X, Y)$ is reduced to $(X' = X \setminus \set{x}, Y' = Y \setminus \set{y})$.
    The reduced case is in Case~B because $X'$ and $Y'$ are contained in different co-components.
    Therefore, $\dist{\CSone}{X}{Y} \le 1 + \size{X' \setminus Y'} = |X \setminus Y|$.
    By $\dist{\CSone}{X}{Y} \ge |X \setminus Y|$, we obtain $\dist{\CSone}{X}{Y} = |X \setminus Y|$.

    \emph{Case~C-1}: $C_X = C_Y$ and $X$ touches $Y$.
    Since $X$ touches $Y$, we obtain $\dist{\CSone}{X}{Y} = |X \setminus Y|$ in the same way as Case~A.

    \emph{Case~C-2}: $C_X = C_Y$ and $X$ does not touch $Y$.
    In this case, there exists another co-component $C_Z$ because $|V(G)| \ge 2$.
    We move any $x \in X$ to any vertex $z \in C_Z$.
    Then, the problem is reduced to $(X'' = (X \setminus \set{x}) \cup \set{z}, Y)$.
    The reduced case is in Case~A because $X''$ intersects with multiple co-components.
    Thus, $\dist{\CSone}{X}{Y} \le 1 + \size{X'' \setminus Y} = |X \setminus Y| + 1$.
    Since $X$ does not touch $Y$, $\dist{\CSone}{X}{Y} > |X \setminus Y|$ holds.
    Therefore, we obtain $\dist{\CSone}{X}{Y} = |X \setminus Y| + 1$.
    \qed
\end{proof}

\subsection{Pseudocode}
We show pseudocode in \Cref{alg:cograph} and its subroutine in \Cref{alg:one_cc}.\footnote{We intentionally aligned the representation of \Cref{alg:cograph} with Kami{\'n}ski et al.~\cite[Algorithm~1]{kaminski2012complexity}. The readers are encouraged to compare ours with theirs.}
In \Cref{alg:cograph}, a \CS-path refers to a reconfiguration sequence under \CS.
\Cref{alg:one_cc} shows a subroutine for computing a shortest reconfiguration sequence under \CSone.
The subroutine for \CS\ can be similarly implemented.

\begin{algorithm}[hbt!]
\DontPrintSemicolon
\caption{\CCR-\CS\ in cographs}\label{alg:cograph}
\KwIn{A cograph $G$ and two vertex subsets $A, B$.}
\KwOut{A shortest sequence from $A$ to $B$ if they are reconfigurable, $\NO$ otherwise.}
\If{$m(A) \neq m(B)$}{\label{line:first}
    \Return{$\NO$}\;
}
\uElseIf{$\size{V(G)} = 1$}{
    \Return{the trivial \CS-path}\;
}
\uElseIf{$G$ is disconnected}{
    let $C_1, \dots, C_k$ be the connected components of $G$\;
    solve the problem recursively for the connected components $C_1, \dots, C_k$ with respective vertex subsets $(A \cap C_1, B \cap C_1), \dots, (A \cap C_k, B \cap C_k)$\;
    \uIf{one of the outputs is \NO}{
        \Return{\NO}\;
    }
    \Else{
        merge corresponding $(A \cap C_i, B \cap C_i)$ paths into an $(A, B)$ \CS-path $P$\;
        \Return{$P$}\;
    }
}
\uElseIf{$\size{\C{A}} = \size{\C{B}} = 1$}{\label{line:one_cc_start}
    \Return{$\textsc{CSOneComponent}(G, A, B)$}\;\label{line:subroutine}
}\label{line:one_cc_end}
\uElseIf{$A$ and $B$ are in the same co-component of $G$}{\label{line:same_cocomp_start}
    solve the problem for $A$ and $B$ recursively on that co-component and\;
    \Return{the output}\;
}\label{line:same_cocomp_end}
\Else{
    \Return{\NO}\;
}
\end{algorithm}

\begin{algorithm}[hbt!]
\DontPrintSemicolon
\caption{$\textsc{CSOneComponent}(G, A, B)$}\label{alg:one_cc}
\KwIn{A connected cograph $G\ (\size{V(G)} \ge 2)$ and two connected vertex subsets $A, B$ with $|A| = |B|$.}
\KwOut{A shortest sequence from $A$ to $B$ if they are reconfigurable, $\NO$ otherwise.}
let $\C{A} = \set{X}$ and $\C{B} = \set{Y}$\;
let $P$ be an empty sequence\;
\If{$X$ does not touch $Y$}{
    let $C_{XY}$ be the co-component containing $X$ and $Y$\;
    let $C_Z$ be another co-component\;
    Take arbitrarily $x \in C_{XY}$ and $z \in C_Z$\;
    Add ``$\mbox{move } x \mbox{ to } z$'' to $P$\;
    $X \gets (X \setminus \set{x}) \cup \set{z}$\;
}
let $C_X$ and $C_Y$ be the co-components containing $X$ and $Y$, respectively\;
let $x \in X$ be a vertex touching $Y$\;
Move vertices in $X \setminus (\set{x} \cup Y)$ to $Y \setminus X$, and finally move $x$ to $Y \setminus X$ (if $x \notin Y$)\;
\Return{$P$}\;
\end{algorithm}

\section{Omitted proof in \Cref{sec:chordal}}
We show the proof of \Cref{lem:chordal}.
\begin{proof}
    Let $G$ be a chordal graph, $A, B \subseteq V$ be two vertex subsets with the same CC-multiset.
    Assume that $\ccpiran{A}{B}$ contains an induced $C_{2k}$ for some $k \ge 2$.
    Let the components on $C_{2k}$ in $\ccpiran{A}{B}$ be $C_A^1, C_B^1, C_A^2, C_B^2, \dots, C_A^k, C_B^k$.
    For $i \in \set{1, \dots, k}$, $C_A^i$ touches $C_B^{i-1}$ and $C_B^i$, and does not touch the other components. (Here, we consider integers in 1 + mod $k$).
    Thus, for every $i \in \set{1, \dots, k}$, there is a vertex of $G$ that is contained in $C_A^i$ and not contained in the other components.
    The same holds for $C_B^i$'s.
    Let the vertices $v_A^1, v_B^1, v_A^2, v_B^2, \dots, v_A^k, v_B^k$, and $C^*$ be a shortest cycle passing through the vertices in this order.
    Since $C^*$ is shortest, it's an induced cycle in $G$.
    In addition, from $k \ge 2$, the length of $C^*$ is at least four.
    Therefore, $G$ contains an induced $C_{\ell}$ for some $\ell \ge 4$, which contradicts that $G$ is chordal.
    \qed
\end{proof}

\end{document}